\documentclass[conference]{IEEEtran}

\usepackage{tikz}
\usepackage{color}
\usetikzlibrary{arrows, calc, snakes}
\usepackage{cite}
\usepackage{graphicx}
\usepackage[cmex10]{amsmath}
\usepackage{amsfonts}
\usepackage{amssymb}
\usepackage{amsthm}
\usepackage{algorithmic}
\usepackage{array}
\usepackage{mdwmath}
\usepackage{mdwtab}
\usepackage{fixltx2e}
\usepackage{url}
\usepackage{enumerate}
\usepackage{comment}

\usepackage{dblfloatfix}
\usepackage{subfigure}
\usepackage[justification=centering]{caption}

\usepackage{relsize}

\newtheorem{definition}{Definition}
\newtheorem{theorem}{Theorem}
\newtheorem{lemma}{Lemma}

\newtheorem{protocol}{Protocol}

\begin{document}

\title{On the Oblivious Transfer Capacity of the Degraded Wiretapped Binary Erasure Channel}

\author{\IEEEauthorblockN{Manoj Mishra and Bikash Kumar Dey }
\IEEEauthorblockA{IIT Bombay, India\\
Email: \{mmishra,bikash\}@ee.iitb.ac.in}
\and
\IEEEauthorblockN{Vinod M. Prabhakaran}
\IEEEauthorblockA{TIFR, Mumbai, India\\
Email: vinodmp@tifr.res.in}
\and
\IEEEauthorblockN{Suhas Diggavi}
\IEEEauthorblockA{UCLA, USA\\
Email: suhas@ee.ucla.edu}
}


\maketitle

\begin{abstract}
We study oblivious transfer (OT) between Alice and Bob in the presence of an eavesdropper Eve over a degraded wiretapped binary erasure channel from Alice to Bob and Eve. In addition to the privacy goals of oblivious transfer between Alice and Bob, we require privacy of Alice and Bob's private data from Eve. In previous work we derived the OT capacity (in the honest-but-curious model) of the wiretapped binary independent erasure channel where the erasure processes of Bob and Eve are independent. Here we derive a lower bound on the OT capacity in the same secrecy model when the wiretapped binary erasure channel is degraded in favour of Bob.
\end{abstract}

\IEEEpeerreviewmaketitle


\section{Introduction}
\label{sec:intro}

In secure multiparty computation, mutually distrusting users want to collaborate in computing functions of their data. They want to do this in such a way that no user derives additional information about other users' data than the function they compute. This has applications in several areas including data-mining, voting, auctions etc.~\cite{CramDamNiel}. In general, secure computation is not possible between users who only have access to private/common randomness and noiseless communication~\cite{kushilevitz1992}. For two-user secure computation, a noisy channel between the users (in addition to a noise-free public channel) provides a stochastic resource on which secure computation can be based~\cite{CrepKilian1988}. Oblivious transfer (OT), which is a specific two-user secure computation, has been proposed as a primitive on which all secure computation can be based~\cite{jkilian1988, jkilian2000}, and OT itself can be obtained from noisy channels.

In 1-of-2 string OT, one of the users, say, Alice, has two bit-strings of equal
length. The other user, say, Bob, wants to learn exactly one of the two
strings. Bob does not want Alice to find out which of the two strings he wants,
while Alice wants to ensure that Bob does not learn anything more than one of
the strings. The OT capacity of a discrete memoryless channel is the largest
rate of the string-length per channel use that can be achieved. We assume that
the users are honest-but-curious, that is, they follow the protocol agreed upon
but they may try to gain illegitimate information about the other user's
private data from everything they have learned at the end of the protocol. For
such users, Nascimento and Winter~\cite{NascWinter2008} obtained a lower bound
on the OT capacity of noisy channels and distributed sources.  Ahlswede and
Csisz\'{a}r~\cite{ot2007}~obtained lower bounds on the honest-but-curious OT capacity for
generalized erasure channels. These lower bounds are tight when the erasure probability is at least $\frac{1}{2}$.  Pinto et.  al.~\cite{PintoDowsMorozNasc2011}
showed that, for erasure probability at least $\frac{1}{2}$, the OT capacity of generalized erasure channels remains the same even when
the users are \emph{malicious}, that is, even if a dishonest user arbitrarily deviates
from the protocol.

Presence of third parties is a natural concern when using noisy channels.
Motivated by this, OT over wiretapped binary erasure channel was studied
in~\cite{MishraDPDisit14}.  Building on the ideas
from~\cite{NascWinter2008,ot2007}, the OT capacity of this channel was
characterized there for the honest-but-curious model. Both 2-privacy, where the
eavesdropper may collude with either Alice or Bob, and 1-privacy, where there
are no collusions, were considered. In~\cite{MishraDPDitw-invited14}, the
problem of performing independent OTs between Alice and each of the other
parties over a binary erasure broadcast channel was considered. Inner and outer
bounds on the OT capacity region were presented which meet except in one regime
of parameter ranges.

Here we study the OT capacity of the {\em degraded} wiretapped binary erasure
channel.
The problem presents some interesting new features. Oblivious transfer relies
on the noise in the legitimate channel (Alice-to-Bob channel) to hide
information of Alice and Bob from each other. In a degraded erasure channel,
the wiretapper obtains more information about the noise process on the
legitimate channel (compared to an independent erasure channel where it
receives no information on this noise process from the channel). Our achievable
scheme is in fact more involved compared to the one
in~\cite{MishraDPDisit14,MishraDPDitw-invited14} precisely because of this
fact. This is in contrast with secret key agreement using broadcast channels
with public discussion where optimal schemes are simpler for the degraded
channel in the sense that no public discussion is needed to achieve secret key
capacity when the channel is degraded~\cite{Maurer93,AhlswedeCs93}. While we do
not prove the optimality of our scheme, we believe that the additional
complexity is unavoidable.

In Section~\ref{sec:prob_defn}, we present the formal problem definition. The
main result is presented in Section~\ref{sec:result_summary}, and the proof is
presented in Section~\ref{sec:achievability}. 


\section{Problem Statement}
\label{sec:prob_defn}

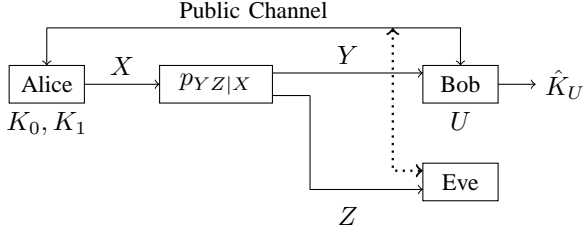
\begin{figure}[h]
\setlength{\unitlength}{1cm}
\centering
\begin{tikzpicture}[scale=1]

\draw (1,3) rectangle (2,3.5);
\draw (3,3) rectangle (4.5,3.5);
\draw (6.5,1.7) rectangle (7.5,2.2);
\draw (6.5,3) rectangle (7.5,3.5);

\draw [->] (2,3.25) -- (3,3.25 );
\draw [->] (4.5,3.4) -- (6.5, 3.4);
\draw (4.5, 3.1) -- (5,3.1);
\draw [->] (5,3.1) |- (6.5,1.85);
\draw [->] (1.5,4) -- (1.5,3.5);
\draw [->] (1.5,4) -| (7,3.5);
\draw [<->, dotted, line width=0.3mm] (6.1,4) |- (6.5, 2.1);

\node at (1.5,3.25) {\small{Alice}};
\node at (7,3.25) {\small{Bob}};
\node at (7,1.95) {\small{Eve}};
\node at (3.75, 3.25) {$p_{YZ|X}$};
\node [above] at (4.25,4) {\small{Public Channel}};
\node [below] at (1.5,3) {$K_0,K_1$};
\node [below] at (7,3) {$U$};

\node [above] at (2.5,3.25) {$X$};
\node [above] at (5.5,3.4) {$Y$};
\node [above] at (5.5,1.25) {$Z$};

\draw [->] (7.5,3.25) -- (8,3.25);
\node [right] at (8,3.25) {$\hat{K}_U$};

\end{tikzpicture}
\caption{Setup for oblivious capacity over a broadcast channel}
\label{fig:ot-setup-bcast}
\end{figure}

In the setup of Fig.~\ref{fig:ot-setup-bcast}, Alice is connected to Bob and an
eavesdropper Eve over a broadcast channel  $p_{YZ|X}$. Additionally, there is a
public channel of unlimited capacity, over which Alice and Bob can take turns
to send messages. Each message sent over this public channel is received by all
users. Alice's private data is a pair of strings $K_0,K_1$ which are $m$-bit
each, while Bob's private data is a choice bit $U$. $K_0,K_1,U$ are independent
and uniform over their respective alphabets. The goal is for Alice and Bob to
do an OT using ($K_0,K_1,U$), without Eve learning anything about
($K_0,K_1,U$).  

We consider the degraded broadcast channel setup shown in Fig.~\ref{fig:ot-setup}.
This is a special case of Fig.~\ref{fig:ot-setup-bcast}.
Here the broadcast channel $p_{YZ|X}$ is made up of a cascade of two independent binary erasure channels (BECs), BEC($\epsilon_1$) with erasure probability $\epsilon_1$ followed by a BEC($\epsilon_2$).
That is, $X\in \{0,1\}, Y,Z\in \{0,1,e\}$, and $p_{YZ|X} = p_{Y|X}p_{Z|Y}$
with $p_{Y|X}(e|1)=p_{Y|X}(e|0)=\epsilon_1$,
$p_{Y|X}(1|1)=p_{Y|X}(0|0)=1- \epsilon_1$, and
$p_{Z|Y}(e|e)=1, p_{Z|Y}(e|1)=p_{Z|Y}(e|0)=\epsilon_2$, 
$p_{Z|Y}(1|1)=p_{Z|Y}(0|0)=1-\epsilon_2$.
\begin{definition}
Let $n,m$ be positive integers. An ($n,m$)-\emph{protocol} is an exchange of messages over the setup of Figure~\ref{fig:ot-setup-bcast}. Alice transmits a bit over the broadcast channel $p_{YZ|X}$ at each time instant $t=1,2,\ldots,n$. Alice's private strings are $m$-bits each. Before each channel use by Alice and also after Alice's last channel use, Alice and Bob can take turns to send an arbitrary but finite number of messages over the public channel. 
\end{definition}
 
We denote by $\mathbf{F}$ the transcript of the public channel at the end of the protocol.
\begin{definition}
The \emph{final view} of a user is the set of random variables it generates and receives over the execution of the protocol. The final views of Alice, Bob and Eve are, respectively,
\begin{align*}
V_A & := (K_0,K_1,X^n,\mathbf{F}), \\
V_B & := (U,Y^n,\mathbf{F}), \\
V_E & := (Z^n,\mathbf{F}).
\end{align*}
\end{definition}

At the end of the protocol, Bob generates an estimate $\hat{K}_U$ of $K_U$, as a function of its final view $V_B$.
\begin{definition}
The $rate$ $r_n$ of an ($n,m$)-protocol is
\begin{equation}
r_n := \frac{m}{n}
\end{equation}
\end{definition}

\begin{definition}
A rate $R$ is \emph{achievable} in the setup of Figure~\ref{fig:ot-setup-bcast} if there exists a sequence of ($n,m$)-protocols such that, as $n \longrightarrow \infty$, $r_n \longrightarrow R$ and
\begin{align}
P[\hat{K}_U \neq K_U] & \longrightarrow 0 \label{eqn:ach_rate_1},\\
I(K_{\overline{U}} ; V_B) & \longrightarrow 0 \label{eqn:ach_rate_2} , \\
I(U ; V_A) & \longrightarrow 0 \label{eqn:ach_rate_3}, \\
I(K_0,K_1,U ; V_E) & \longrightarrow 0 \label{eqn:ach_rate_4} .
\end{align}
\end{definition}

\begin{definition}
The \emph{capacity} $C$ in the setup of Figure~\ref{fig:ot-setup-bcast} is
\begin{equation}
C := \sup\{R : R \text{ is achievable}\}
\end{equation}
\end{definition}

\begin{figure}[h]
\setlength{\unitlength}{1cm}
\centering
\begin{tikzpicture}[scale=1]

\draw (1,3) rectangle (2,3.5);
\draw (3.5,3) rectangle (5,3.5);
\draw (6.5,3) rectangle (8,3.5);
\draw (5.25,1) rectangle (6.25,1.5);
\draw (8.25,1) rectangle (9.25,1.5);

\draw [->] (2,3.25) -- (3.5,3.25);
\draw [->] (5,3.25) -- (6.5,3.25);
\draw [->] (5.5,3.25) -- (5.5,1.5);
\draw [->] (8,3.25) -| (8.5,1.5);

\draw [<-](1.5,3.5) |- (6,4);
\draw [->] (6,4) -- (6,1.5);
\draw [->, dotted, line width=0.3mm] (6,4) -| (9,1.5);
\draw [->] (6.25,1.25) -- (6.75,1.25);

\node at (1.5,3.25) {\small{Alice}};
\node at (5.75,1.25) {\small{Bob}};
\node at (8.75,1.25) {\small{Eve}};
\node at (4.25, 3.25) {\small{BEC($\epsilon_1$)}};
\node at (7.25, 3.25) {\small{BEC($\epsilon_2$)}};
\node [above] at (5.25,4) {\small{Public Channel}};
\node [below] at (1.5,3) {\small{$K_0,K_1$}};
\node [below] at (5.75,1) {$U$};
\node [right] at (6.75,1.25) {$\hat{K}_U$};

\node [above] at (2.75,3.25) {$X$};
\node [above] at (5.5,3.25) {$Y$};
\node [above] at (8.35,3.25) {$Z$};

\end{tikzpicture}
\caption{Setup for oblivious transfer over a degraded binary erasure broadcast channel}
\label{fig:ot-setup}
\end{figure}
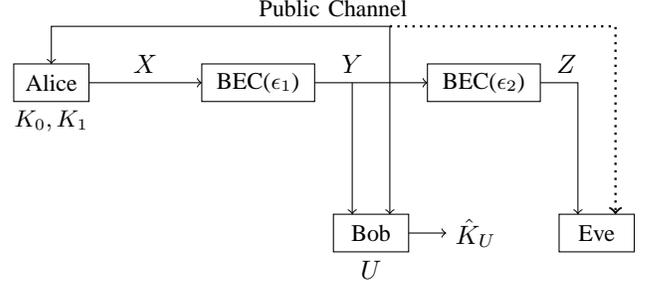


\section{Summary of Results}
\label{sec:result_summary}

Our main result is a lower bound on $C$, for the setup of Figure~\ref{fig:ot-setup}. 

\begin{theorem}
\label{thm:main_result}
\begin{equation*}
\min \left\{\frac{1}{3}\epsilon_2(1 - \epsilon_1), \epsilon_1 \right\} \leq C \leq \min\{\epsilon_2(1 - \epsilon_1), \epsilon_1 \}.
\end{equation*}
\end{theorem}

Hence, when $\epsilon_1 \leq \frac{1}{3}\epsilon_2(1 - \epsilon_1)$, we have $C = \epsilon_1$. Note that the upper bound of $\epsilon_2(1 - \epsilon_1)$ in Theorem~\ref{thm:main_result} follows from the fact that OT capacity is upper bounded by the secret key capacity of the wiretapped channel. This is because if Bob runs the protocol with the choice bit set deterministically to, say, 0, then $K_0$ is a secret key between Alice and Bob. The upper bound follows from the fact that $\epsilon_2(1-\epsilon_1)$ is the secret key capacity of this wiretapped channel with public discussion~\cite{Maurer93,AhlswedeCs93}. The upper bound of $\epsilon_1$ follows from fact that this is an upper bound for two-party OT capacity of the binary erasure channel with erasure probability $\epsilon_1$~\cite{ot2007}.


Section~\ref{sec:achievability} gives a protocol that is used to prove the lower bound part of Theorem~\ref{thm:main_result}.


\section{Achievability of Theorem~\ref{thm:main_result}}
\label{sec:achievability}

In this section, we will describe a protocol which will be used to show that the lower bound in Theorem~\ref{thm:main_result} is an achievable rate. We begin by presenting the main ideas used in this protocol, before presenting a formal description for it.

Let $r$ be any rate smaller than $\min\{ \frac{1}{3}\epsilon_2 (1-\epsilon_1), \epsilon_1 \}$.
In the protocol, Alice begins by transmitting a sequence $X^n$ of
independent bits, each equally likely to be $0$ or $1$. Bob receives the
corresponding erased version $Y^n$, while Eve gets $Z^n$ which is an erased
version of $Y^n$ (see Figure~\ref{fig:ot-setup}). Let $E$ denote the set of
\emph{indices} at which $Y^n$ has erasures, and let $\overline{E}$
be its complement. 
Out of the set $\overline{E}$, Bob picks a \emph{good set} $G$ of size
$|G| = nr$ uniformly at random. 
Similarly, Bob also picks a \emph{bad set} $B$ of size $|B|=nr$ uniformly at
random out of set $E$. Let $\tilde{G}:=\overline{E}\setminus G$ denote
the unerased indices which are not in $G$. Note that, $|\tilde{G}|$ is 
approximately $n(1 - \epsilon_1 - r)$. Similarly, $\tilde{B}:=E\setminus B$ is the 
set of erased indices which are not in $B$, and $|\tilde{B}|$ is approximately
$n(\epsilon_1 - r)$. If $U=0$, Bob assigns sets $(L_0,L_1) = (G,B)$, 
otherwise Bob sets $(L_0,L_1) = (B,G)$. 

Bob first declares the sets $\tilde{G},\tilde{B}$ over the public channel. Thereafter, Bob forms sequential, disjoint subsets ($\tilde{G}_L, \tilde{G}_S$) of the set $\tilde{G}$, where $|\tilde{G_L}|$ is about $\frac{2nr}{\epsilon_2}$ and $|\tilde{G}_S|$ is about $\frac{nr(1 - \epsilon_2)}{\epsilon_2}$. Both $\tilde{G}_L$ and $\tilde{G}_S$ will be used to generate secret keys known to both Alice and Bob, but hidden from Eve. Note that for $r < \min\{ \frac{1}{3}\epsilon_2(1 - \epsilon_1), \epsilon_1 \}$, all these sets of the required sizes can be formed.

Alice and Bob use the set of transmissions $X^n|_{\tilde{G}_L}$ to agree on a secret key $S_L$, secret from Eve. The size of $S_L$ is $2nr$ bits. Similarly, Alice and Bob use $X^n|_{L_0 \cup \tilde{G}_S}$ to form a secret key $S_0$ and $X^n|_{L_1 \cup \tilde{G}_S}$ to form a secret key $S_1$. Here, $|S_0| = |S_1| = nr$. 

Bob now needs to reveal ($L_0,L_1$) to Alice, while hiding both these sets from Eve, so as not to reveal $U$ to Eve. Towards this goal, Bob forms a set $L$ which is an ordered version of the set $L_0 \cup L_1$. Then, Bob forms a binary vector $Q$ of length $2nr$ as follows. For all $i=0,1,\ldots,2nr-1$,  the $i$th element of $Q$, denoted by $Q_i$ will indicate whether the $i$th element of $L$, denoted by $L_i$, belongs to $L_0$ or $L_1$. That is, $Q_i = 0$ when $L_i \in L_0$, otherwise $Q_i = 1$. Bob now sends $S_L \oplus Q$ to Alice, which is sufficient for Alice to recover ($L_0,L_1$).

Alice finally sends the encrypted strings $K_0 \oplus S_0$ and $K_1 \oplus S_1$ to Bob over the public channel. Since Bob knows $S_U$ completely (since he knows $X^n|_{L_U \cup \tilde{G}_S}$), Bob can recover $K_U$.

Before presenting the protocol more formally, we point out a comparison with the independent erasure case of~\cite{MishraDPDisit14}. In the independent erasure channel, since Eve has no side information about the erasure pattern of Bob, there is no need for Bob to encrypt $(L_0,L_1)$ before sending it over the public channel. The additional burden of encrypting $L_0,L_1$ here requires Alice and Bob to generate a secret key $S_L$ which is twice as long as each string $K$. Along with $S_U$, effectively, Alice and Bob agree on secret keys which are together thrice as long as each string $K$. This explains the $\frac{1}{3}$ factor in the rate achieved. 

\begin{protocol}
\label{sch:achievable}

Let $\delta \in (0,1)$. Let $\tilde{\epsilon}_2 = \epsilon_2(1 - \delta)$. Let $r = \min\{ \frac{1}{3}\tilde{\epsilon}_2(1 - \epsilon_1) - \theta_{\delta},  \epsilon_1 - \delta\}$ be the rate to be achieved, where $\theta_{\delta} = \delta(1 + \frac{2}{\tilde{\epsilon}_2})$.
\begin{description}

\item[\textbf{Alice}] Transmits a sequence $X^n$ of independent bits, equally likely to be $0$ or $1$, over the broadcast channel in Figure~\ref{fig:ot-setup}.

\item[\textbf{Bob}] Receives the sequence $Y^n$ from the output of BEC($\epsilon_1$) and forms the erased and unerased sets of indices of $Y^n$ as, respectively,

\begin{align*}
E & := \{i \in \{1,2,\ldots,n\} : Y_i = e \} \\
\overline{E} & := \{i \in \{1,2,\ldots,n\} : Y_i \neq e \}
\end{align*}
If $|\overline{E}| < n(1 - \epsilon_1 - \delta)$ or $|E| < n(\epsilon_1 - \delta)$, Bob declares an error and quits. Otherwise, Bob proceeds and forms the following sets out of $E$ and $\overline{E}$:

\begin{align*}
G & := \text{Unif}\left\{ A \subset \overline{E} : |A| = n(r + \delta) \right\} \\
B & := \text{Unif}\left\{ A \subset E : |A| = n(r + \delta) \right\} \\
\tilde{G} & := \overline{E} \backslash G \\
\tilde{B} & := E \backslash B
\end{align*}

where $\text{Unif}\{.\}$ denotes a random, uniformly distributed, 
choice over the collection of sets.
Note that $|\tilde{G}| \geq n(1 - \epsilon_1 - r - \delta)$ and $|\tilde{B}| \geq n(\epsilon_1 - r - \delta)$.

Bob reveals the sets $\tilde{G}, \tilde{B}$ over the public channel. Bob now forms the sets $L_0,L_1$ as follows:

\begin{align*}
U = 0:\qquad &L_0 = G,\quad L_1 = B \\
U = 1:\qquad &L_0 = B,\quad L_1 = G
\end{align*}

Let $L$ be an ordered version of the set $L_0 \cup L_1$. Bob forms a binary vector $Q$ of $2nr$ bits, with elements labelled $Q_i$, $i=0,1,\ldots,2nr-1$ defined as:

\begin{equation*}
Q_i = \left\{ \begin{array}{ll} 0, & L_i \in L_0 \\ 1, & L_i \in L_1 \end{array} \right.
\end{equation*}

Bob takes the first $\frac{2n(r + \delta)}{\tilde{\epsilon}_2}$ indices in $\tilde{G}$ and calls it set $\tilde{G}_L$. Bob also takes the next $\frac{nr(1 - \tilde{\epsilon}_2)}{\tilde{\epsilon}_2}$ indices in $\tilde{G}$ and calls it set $\tilde{G}_S$. One can verify that $|\tilde{G}_L| + |\tilde{G}_S| \leq |\tilde{G}|$. Bob then forms a secret key $S_L$ using $X^n|_{\tilde{G}_L}$, where $S_L$ is known to Alice but hidden from Eve. Here, $S_L$ is $2n(r + \delta)$ bits long. Finally, Bob sends the following quantity to Alice over the public channel:

\begin{equation*}
S_L \oplus Q  
\end{equation*}  

\item[\textbf{Alice}] Uses $X^n|_{L_0 \cup \tilde{G}_S}$ to form a secret key $S_0$ and uses $X^n|_{L_1 \cup \tilde{G}_S}$ to form a secret key $S_1$. Both $S_0,S_1$ are $nr$-bit each. Alice finally sends the following two encrypted strings to Bob over the public channel:

\begin{align*}
K_0 & \oplus S_0 \\
K_1 & \oplus S_1
\end{align*}

\item[\textbf{Bob}] Knows $X^n|_{L_U \cup \tilde{G}_S}$ and hence knows $S_U$, thereby recovering $K_U$ from Alice's public message.

\end{description}

\end{protocol}

\begin{lemma}
\label{lem:achievability}
A rate of $\min \left\{ \frac{1}{3}\epsilon_2(1 - \epsilon_1), \epsilon_1 \right\}$ is achievable in the setup of Figure~\ref{fig:ot-setup}.
\end{lemma}

A proof of this lemma is deferred to the Appendix. Below, we give a sketch of this proof.


\begin{itemize}

\item (\ref{eqn:ach_rate_1}) is satisfied for the following reason. Since Bob knows both $X^n|_{L_U}$ and $X^n|_{\tilde{G}_S}$, Bob knows $X^n|_{L_U \cup \tilde{G}_S}$. Hence, Bob knows the secret key $S_U$ and so, Bob can recover $K_U$ correctly from $K_U \oplus S_U$ that Alice sends on the public channel.

\item (\ref{eqn:ach_rate_2}) is satisfied because Bob knows nothing about $X^n|_{L_{\overline{U}}}$. Since $S_{\overline{U}}$ is a secret key generated from $X^n|_{L_{\overline{U}} \cup \tilde{G}_S}$ and has the same number of bits as $X^n|_{L_{\overline{U}}}$, Bob will learn practically no information about $S_{\overline{U}}$ and, hence, about $K_{\overline{U}}$.

\item Alice can learn about $U$ only from Bob's public messages. In the scheme, Alice learns $L_0,L_1$ from Bob's public messages. Since $L_0,L_1$ are of the same size and since the channel acts independently on each input bit, Alice learns no information about $U$. Hence, (\ref{eqn:ach_rate_3}).

\item 
Finally, Eve cannot learn $U$ since the identity of $L_0,L_1$ remains hidden from her by the secret key $S_L$. Eve only learns $L_0 \cup L_1$ and nothing more. Conditioned on knowing $U$, Eve still does not learn $(K_0,K_1)$ since these are encrypted using secret keys $S_0,S_1$ which are secret from Eve. Hence, (\ref{eqn:ach_rate_4}) is satisfied.
\end{itemize}


\section{Acknowledgements}
The work  was supported in part by the Bharti Centre for Communication, IIT Bombay, a grant from the Department of Science and Technology, Government of India, to IIT Bombay, and by Information Technology Research Academy (ITRA), Government of India under ITRA-Mobile grant ITRA/15(64)/Mobile/USEAADWN/01.  V. Prabhakaran's research was also supported in part by a Ramanujan Fellowship from the Department of Science and Technology, Government of India. The work of S. Diggavi was supported in part by NSF grant 1321120.



\appendices

\section{Proof of Lemma~\ref{lem:achievability}}

In order to prove Lemma~\ref{lem:achievability}, we will use a sequence $\{P_n\}_{n \in \mathbb{N}}$ of Protocol~\ref{sch:achievable} and show that for $r < \min \left\{\frac{1}{3}\epsilon_2(1 - \epsilon_1), \epsilon_1 \right\}$, (\ref{eqn:ach_rate_1}) - (\ref{eqn:ach_rate_4}) hold for $\{P_n\}_{n \in \mathbb{N}}$. 

We note that for $P_n$, the transcript of the public channel is
\begin{equation}
\mathbf{F} = \text{(} \tilde{G}, \tilde{B}, S_L \oplus Q, K_0 \oplus S_0, K_1 \oplus S_1 \text{)}.
\end{equation}

Let $J$ be the indicator random variable for the event that Bob declares an error and quits. Using Chernoff bound, we see that $P[J=1] \longrightarrow 0$ as $n \longrightarrow \infty$.

\begin{enumerate}

\item \label{part:1} In order to show (\ref{eqn:ach_rate_1}) holds, given that $P[J=1] \longrightarrow 0$, it suffices to show that $P[\hat{K}_U \neq K_U | J = 0] \longrightarrow 0$.

When $J=0$, Bob knows $X^n|_{L_U}$ and Bob also knows $X^n|_{\tilde{G}_S}$. Hence, Bob knows $X^n|_{L_U \cup \tilde{G}_S}$. As a result, Bob knows the secret key $S_U$ derived out of $X^n|_{L_U \cup \tilde{G}_S}$. Hence, Bob can get $K_U$ using $K_U \oplus S_U$ sent by Alice. Thus, $P[\hat{K}_U \neq K_U | J = 0] = 0$.

\item \label{part:2} In order to show (\ref{eqn:ach_rate_2}) holds, it will suffice to show that $I(K_{\overline{U}} ; V_B | J=0) \longrightarrow 0$. All terms and assertions below are conditioned on the event $J=0$, but this is being suppressed for ease of writing.

\vspace{-0.5cm}

\begin{align*}
& I(K_{\overline{U}} ; V_B) \\
&= I(K_{\overline{U}} ; U, Y^n, \mathbf{F}) \\
&= I(K_{\overline{U}} ; U,Y^n, \tilde{G}, \tilde{B}, S_L \oplus Q, K_0 \oplus S_0, K_1 \oplus S_1) \\
&= I(K_{\overline{U}} ; U,Y^n, \tilde{G}, \tilde{B}, S_L \oplus Q, K_U \oplus S_U, K_{\overline{U}} \oplus S_{\overline{U}}) \\
&= I(K_{\overline{U}} ; U,Y^n, \tilde{G}, \tilde{B}, Q, K_U \oplus S_U, K_{\overline{U}} \oplus S_{\overline{U}}) \\
& \text{[since $S_L$ is a function of ($Y^n,\tilde{G}$)]} \\
&= I(K_{\overline{U}} ; U,Y^n, \tilde{G}, \tilde{B}, G,B, K_U \oplus S_U, K_{\overline{U}} \oplus S_{\overline{U}}) \\
& \text{[since ($G,B$) is a function of ($U,Q,\tilde{G},\tilde{B}$) and $Q$ is} \\ & \text{ a function of ($U,G,B$)]} \\
&= I(K_{\overline{U}} ; U,Y^n, \tilde{G}, \tilde{B}, G,B, K_U, K_{\overline{U}} \oplus S_{\overline{U}}) \\
& \text{[since $S_U$ is a function of ($Y^n,G,\tilde{G}$)]} \\
&= I(K_{\overline{U}} ; U,Y^n, \tilde{G}, \tilde{B}, G,B, K_{\overline{U}} \oplus S_{\overline{U}}) \\
& \text{[since $K_U$ is independent of all other variables above]} \\
&= I(K_{\overline{U}} ;  K_{\overline{U}} \oplus S_{\overline{U}} | U,Y^n,G,B,\tilde{G},\tilde{B}) \\
&= I(K_{\overline{U}} ;  K_{\overline{U}} \oplus S_{\overline{U}}    \;\; \mathlarger{\mid} \; \;     U,Y^n,Y^n|_{\tilde{G}_S} G,B,\tilde{G},\tilde{B}) \\
&= I(K_{\overline{U}} ;  K_{\overline{U}} \oplus S_{\overline{U}}  \;\; \mathlarger{\mid} \; \; Y^n|_{\tilde{G}_S}) \\
& \text{[since $S_{\overline{U}} - Y^n|_{\tilde{G}_S} - U,Y^n,G,B,\tilde{G},\tilde{B}$ is a}\\& \text{ Markov chain]} \\
&= I(K_{\overline{U}} ;  K_{\overline{U}} \oplus S_{\overline{U}} \;\; \mathlarger{\mid} \; \; Y^n|_{\tilde{G}_S}) \\
&= I(K_{\overline{U}} ;  K_{\overline{U}} \oplus S_{\overline{U}}  \;\; \mathlarger{\mid} \; \;  X^n|_{\tilde{G}_S}) \\
&= H(K_{\overline{U}} \oplus S_{\overline{U}}  \;\; \mathlarger{\mid} \; \; X^n|_{\tilde{G}_S}) - H(S_{\overline{U}} \;\; \mathlarger{\mid} \; \; K_{\overline{U}}, X^n|_{\tilde{G}_S}) \\
&= H(K_{\overline{U}} \oplus S_{\overline{U}} \;\; \mathlarger{\mid} \; \; X^n|_{\tilde{G}_S}) - H(S_{\overline{U}} \;\; \mathlarger{\mid} \; \; X^n|_{\tilde{G}_S}) \\
& \leq |S_{\overline{U}}| -  H(S_{\overline{U}} \;\; \mathlarger{\mid} \; \; X^n|_{\tilde{G}_S})
\end{align*}
As a consequence of Lemma~\ref{lem:major} of Appendix~\ref{sec:useful_lemma}, the above quantity is small.

\item \label{part:3} In order to show (\ref{eqn:ach_rate_3}) holds, it suffices to show that $I(U ; V_A | J=0) \longrightarrow 0$. All terms and assertions below are conditioned on the event $J=0$, but this is being suppressed for ease of writing.

\vspace{-0.5cm}

\begin{align*}
& I(U ; V_A ) \\
&= I(U ;K_0,K_1,X^n,\mathbf{F} ) \\
&= I(U ; K_0,K_1,X^n, \tilde{G}, \tilde{B}, S_L \oplus Q, K_0 \oplus S_0, K_1 \oplus S_1) \\
&= I(U ; K_0, K_1, X^n, \tilde{G}, \tilde{B}, S_L \oplus Q, S_0, S_1) \\
&=  I(U ; K_0,K_1, X^n, \tilde{G}, \tilde{B}, Q, S_0,S_1) \\
& \text{[since $S_L$ is a function of ($X^n,\tilde{G}$)]} \\ 
&= I(U ; K_0,K_1,X^n,\tilde{G},\tilde{B}, L_0,L_1,S_0,S_1) \\
& \text{[since ($L_0,L_1$) is a function of ($\tilde{G},\tilde{B},Q$) ]} \\
&= I(U ; K_0,K_1,X^n,\tilde{G},\tilde{B},L_0,L_1) \\
& \text{[since ($S_0,S_1$ is a function of ($X^n,L_0,L_1,\tilde{G}$))]} \\
&= I(U ; L_0,L_1) \\
& \text{[since $U - L_0,L_1 - K_0,K_1,X^n,\tilde{G},\tilde{B}$ is a Markov chain]} \\
&= 0 \\
& \text{[since the channel acts independently on each input bit} \\ &  \text{ and since $|L_0| = |L_1|$]}
\end{align*}

\item \label{part:4} In order to show (\ref{eqn:ach_rate_4}) holds, it will suffice to show that $I(K_0,K_1,U ; V_E | J=0) \longrightarrow 0$ as $n \longrightarrow \infty$. All terms and assertions below are conditioned on the event $J=0$, but this is being suppressed for ease of writing.
\begin{align*}
I(K_0&,K_1,U ; V_E) \\
&= I(K_U, K_{\overline{U}}, U ; V_E ) \\
&= I(U ; V_E) + I(K_{\overline{U}} ; V_E |U) + I(K_U ; V_E | U, K_{\overline{U}}) \\
&= I(U ; V_E ) + I(K_{\overline{U}} ; U, V_E ) + I(K_U ; U, K_{\overline{U}}, V_E)
\end{align*}
We will look at each of the above three terms separately.
\begin{align*}
& I(U ; V_E ) \\
& = I(U ; Z^n, \mathbf{F}) \\
&= I(U ; Z^n, \tilde{G}, \tilde{B}, S_L \oplus Q, K_0 \oplus S_0, K_1 \oplus S_1) \\
&\leq I(U ; Z^n, \tilde{G}, \tilde{B}, S_L \oplus Q, K_0, S_0, K_1, S_1) \\
&= I(U ; Z^n, \tilde{G}, \tilde{B}, S_L \oplus Q, S_0, S_1) \\
& \text{[ since $K_0,K_1$ are independent of all the variables above]} \\
&= I(U ; S_L \oplus Q, S_0, S_1 | Z^n, \tilde{G}, \tilde{B}) \\
&= H(S_L \oplus Q, S_0, S_1 | Z^n, \tilde{G}, \tilde{B})) \\ & \quad - H(S_L \oplus Q, S_0, S_1 | U, Z^n, \tilde{G}, \tilde{B}) \\
& \leq |S_L| + |S_0| + |S_1| -  H(S_L, S_0, S_1 | U, Q, Z^n, \tilde{G}, \tilde{B}) \\
& = |S_L| + |S_U| + |S_{\overline{U}}| -  H(S_L, S_C, S_{\overline{C}} | U, G,B, Z^n, \tilde{G}, \tilde{B}) \\
& \text{[since ($G,B$) is a function of ($U,Q,\tilde{G},\tilde{B}$)  and $Q$ is} \\ & \text{ a function of ($U,G,B$)]} \\
&=  |S_L| + |S_U| + |S_{\overline{U}}| \\ & \quad -  H(S_L, S_U, S_{\overline{U}}  \;\; \mathlarger{\mid} \; \;  Z^n|_G, Z^n|_{\tilde{G}_S}, Z^n|_{\tilde{G}_L}) \\
& \text{[since $S_L, S_U, S_{\overline{U}} - Z^n|_G, Z^n|_{\tilde{G}_S}, Z^n|_{\tilde{G}_L} -$}\\ 
  &\qquad\qquad\qquad\qquad U, G,B, Z^n, \tilde{G}, \tilde{B} \text{ is a Markov chain ]} \\
&= |S_L| + |S_U| + |S_{\overline{U}}|-  H(S_L \;\; \mathlarger{\mid} \; \; Z^n|_G, Z^n|_{\tilde{G}_S}, Z^n|_{\tilde{G}_L}) \\ & \qquad - H(S_U, S_{\overline{U}} \;\; \mathlarger{\mid} \; \; S_L, Z^n|_G, Z^n|_{\tilde{G}_S}, Z^n|_{\tilde{G}_L}) \\
&=  (|S_L|  -  H(S_L \;\; \mathlarger{\mid} \; \; Z^n|_{\tilde{G}_L})) \\
   & \qquad + (|S_U| + |S_{\overline{U}}| - H(S_U, S_{\overline{U}} \;\; \mathlarger{\mid} \; \; Z^n|_G, Z^n|_{\tilde{G}_S})) \\
& \text{[since $S_L - Z^n|_{\tilde{G}_L} -  Z^n|_G, Z^n|_{\tilde{G}_S}$ and}\\
&  S_U, S_{\overline{U}} - Z^n|_G, Z^n|_{\tilde{G}_S} - S_L, Z^n|_{\tilde{G}_L} \text{ are Markov Chains]}
\end{align*}
The first term above is small since $S_L$ is a secret key against Eve. Lemma~\ref{lem:major} of Appendix~\ref{sec:useful_lemma} implies that the second term is also small.
\begin{align*}
&  I(K_{\overline{U}} ; U, V_E) \\
&= I(K_{\overline{U}} ; U, Z^n, \mathbf{F}) \\
&= I(K_{\overline{U}} ; U,  Z^n, \tilde{G}, \tilde{B}, S_L \oplus Q, K_0 \oplus S_0, K_1 \oplus S_1) \\
&= I( K_{\overline{U}} ; U,  Z^n, \tilde{G}, \tilde{B}, S_L \oplus Q, K_U \oplus S_U, K_{\overline{U}} \oplus S_{\overline{U}} ) \\
&\leq I( K_{\overline{U}} ; U,  Z^n, \tilde{G}, \tilde{B}, S_L \oplus Q, K_U, S_U, K_{\overline{U}} \oplus S_{\overline{U}} ) \\
&= I( K_{\overline{U}} ; U,  Z^n, \tilde{G}, \tilde{B}, S_L \oplus Q, S_C, K_{\overline{U}} \oplus S_{\overline{U}} ) \\
& \text{[since $K_U$ is independent of all other variables above]} \\
& \leq I( K_{\overline{U}} ; U,  Z^n, \tilde{G}, \tilde{B}, S_L, Q, S_U, K_{\overline{U}} \oplus S_{\overline{U}} ) \\
&= I( K_{\overline{U}} ; U,  Z^n, \tilde{G}, \tilde{B}, S_L, G,B, S_U, K_{\overline{U}} \oplus S_{\overline{U}} ) \\
& \text{[since ($G,B$) is a function of ($U,Q,\tilde{G},\tilde{B}$)  and $Q$ is} \\ & \text{ a function of ($U,G,B$)]} \\
&= I( K_{\overline{U}} ; S_U,S_L, Z^n|_G, Z^n|_{\tilde{G}}, K_{\overline{U}} \oplus S_{\overline{U}} ) \\
& \text{[since $K_{\overline{U}} - S_U,S_L, Z^n|_G, Z^n|_{\tilde{G}}, K_{\overline{U}} \oplus S_{\overline{U}}$ - } \\ & \text{ $\quad U,Z^n,G,B,\tilde{G},\tilde{B}$ is a Markov chain]} \\
&= I( K_{\overline{U}} ; S_U,S_L, Z^n|_G, Z^n|_{\tilde{G}_L}, Z^n|_{\tilde{G}_S}, K_{\overline{U}} \oplus S_{\overline{U}} ) \\
&= I( K_{\overline{U}} ; S_U, Z^n|_G, Z^n|_{\tilde{G}_S}, K_{\overline{U}} \oplus S_{\overline{U}} ) \\
& \text{[ since $K_{\overline{U}} -  K_{\overline{U}} \oplus S_{\overline{U}}, S_U, Z^n|_G, Z^n|_{\tilde{G}_S} - S_L, Z^n|_{\tilde{G}_L}$} \\ & \text{ is a Markov chain ]} \\
&= I( K_{\overline{U}} ; S_U, K_{\overline{U}} \oplus S_{\overline{U}} \;\; \mathlarger{\mid} \; \; Z^n|_G, Z^n|_{\tilde{G}_S} ) \\
&= H( S_U, K_{\overline{U}} \oplus S_{\overline{U}} \;\; \mathlarger{\mid} \; \; Z^n|_G, Z^n|_{\tilde{G}_S}) \\ & \quad - H(S_U, S_{\overline{U}} \;\; \mathlarger{\mid} \; \; K_{\overline{U}}, Z^n|_G, Z^n|_{\tilde{G}_S}) \\
&=  H( S_U, K_{\overline{U}} \oplus S_{\overline{U}} \;\; \mathlarger{\mid} \; \; Z^n|_G, Z^n|_{\tilde{G}_S}) \\ & \quad - H(S_U, S_{\overline{U}} \;\; \mathlarger{\mid} \; \; Z^n|_G, Z^n|_{\tilde{G}_S}) \\
&\leq |S_U| + |S_{\overline{U}}| - H(S_U, S_{\overline{U}} \;\; \mathlarger{\mid} \; \; Z^n|_G, Z^n|_{\tilde{G}_S})
\end{align*}
The above term is small as a consequence of Lemma~\ref{lem:major} in Appendix~\ref{sec:useful_lemma}.
\begin{align*}
& I(K_U ; C, K_{\overline{U}}, V_E) \\
&= I(K_U ; U, K_{\overline{U}}, Z^n, \mathbf{F}) \\
&= I(K_U ; U, K_{\overline{C}}, Z^n, \tilde{G}, \tilde{B}, S_L \oplus Q, K_0 \oplus S_0, K_1 \oplus S_1) \\
&= I(K_U ; U, K_{\overline{U}}, Z^n, \tilde{G}, \tilde{B}, S_L \oplus Q, K_U \oplus S_U, K_{\overline{U}} \oplus S_{\overline{U}}) \\
&= I(K_U ; U, K_{\overline{U}}, S_{\overline{U}}, Z^n, \tilde{G}, \tilde{B}, S_L \oplus Q, K_U \oplus S_U) \\
&= I(K_U ; U, S_{\overline{U}}, Z^n, \tilde{G}, \tilde{B}, S_L \oplus Q, K_U \oplus S_U ) \\
& \text{[since $K_{\overline{U}}$ is independent of all other varables above]} \\
& \leq I(K_U ; U, S_{\overline{U}}, Z^n, \tilde{G}, \tilde{B}, S_L, Q, K_U \oplus S_U ) \\
&= I(K_U ; U, S_{\overline{U}}, S_L, Z^n, \tilde{G}, \tilde{B}, G,B, K_U \oplus S_U) \\
&  \text{[since ($G,B$) is a function of ($U,Q,\tilde{G},\tilde{B}$)  and $Q$ is} \\ & \text{ a function of ($U,G,B$)]} \\
&= I(K_U ; S_{\overline{U}}, S_L, Z^n|_G, Z^n|_{\tilde{G}}, K_U \oplus S_U) \\
& \text{[since $K_U - S_{\overline{U}}, S_L, Z^n|_G, Z^n|_{\tilde{G}}, K_U \oplus S_U -$} \\ & \text{$\quad  U,Z^n, G,B, \tilde{G}, \tilde{B}$ is a Markov chain]} \\
&= I(K_U ; S_{\overline{U}}, S_L, Z^n|_G, Z^n|_{\tilde{G}_L}, Z^n|_{\tilde{G}_S}, K_U \oplus S_U) \\
&= I(K_U ; S_{\overline{U}}, Z^n|_G, Z^n|_{\tilde{G}_S}, K_U \oplus S_U) \\
& \text{[ since $K_U -  K_U \oplus S_U, S_{\overline{U}}, Z^n|_G, Z^n|_{\tilde{G}_S} - S_L, Z^n|_{\tilde{G}_L}$ } \\ & \text{  is a Markov chain]} \\
&= I(K_U ; S_{\overline{U}}, K_U \oplus S_U \;\; \mathlarger{\mid} \; \; Z^n|_G, Z^n|_{\tilde{G}_S}) \\
&= H( S_{\overline{U}}, K_U \oplus S_U \;\; \mathlarger{\mid} \; \; Z^n|_G, Z^n|_{\tilde{G}_S}) \\& \qquad - H( S_{\overline{U}}, S_U \;\; \mathlarger{\mid} \; \; K_U, Z^n|_G, Z^n|_{\tilde{G}_S}) \\
&=  H( S_{\overline{U}}, K_U \oplus S_U \;\; \mathlarger{\mid} \; \; Z^n|_G, Z^n|_{\tilde{G}_S}) \\ & \qquad - H( S_{\overline{U}}, S_U \;\; \mathlarger{\mid} \; \; Z^n|_G, Z^n|_{\tilde{G}_S}) \\
& \leq |S_{\overline{U}}| + |S_U| - H( S_{\overline{U}}, S_U \;\; \mathlarger{\mid} \; \; Z^n|_G, Z^n|_{\tilde{G}_S}) 
\end{align*} 
The above term is small, again as a consequence of Lemma~\ref{lem:major}.

\end{enumerate}

\section{A Useful Lemma} \label{sec:useful_lemma}
\begin{lemma}
\label{lem:major}
Let Alice transmit a sequence $X^n$ of i.i.d. Ber($\frac{1}{2}$) bits and let Bob define the sets $G,B,\tilde{G}_S$ as in Protocol~\ref{sch:achievable}. There exists a sequence of codes $\{\mathcal{M}_n^U,\mathcal{M}_n^{\overline{U}}\}_{n \in \mathbb{N}}$, where 
$\mathcal{M}_n^U: (X^n|_{\tilde{G}_S},X^n|_G) \mapsto S_U$ and $\mathcal{M}_n^{\overline{U}}: (X^n|_{\tilde{G}_S},X^n|_B) \mapsto S_{\overline{U}}$ such that
 $|S_U| = |S_{\overline{U}}| = nr$ and, for $n \longrightarrow \infty$, 
\begin{enumerate}
\item $H(S_U,S_{\overline{U}} \;\; \mathlarger{\mid} \; \; Z^n|_G,  Z^n|_{\tilde{G}_S} ) -2nr \longrightarrow 0$,
\item $H(S_{\overline{U}} \;\; \mathlarger{\mid} \; \; X^n|_{\tilde{G}_S}) -nr \longrightarrow 0$.
\end{enumerate}
\end{lemma}
\begin{proof}
We prove this lemma using a random coding argument. The codes (maps) $\mathcal{M}_n^U$ and $\mathcal{M}_n^{\overline{U}}$ are chosen independently and uniformly at random from among all possible maps. We treat the entropies $H(S_U,S_{\overline{U}} \; \mathlarger{\mid}  \; Z^n|_G,  Z^n|_{\tilde{G}_S} )$ and $H(S_{\overline{U}} \; \mathlarger{\mid}  \; X^n|_{\tilde{G}_S})$ as random variables (which depend on the random code). It then suffices to show that with high probability they approach $2nr$ and $nr$, respectively.

We will make use of Lemma~\ref{lem:minor}, which is stated after this proof. Consider,
\begin{align*}
&H(S_U,S_{\overline{U}} \; \mathlarger{\mid}  \; Z^n|_G,  Z^n|_{\tilde{G}_S} )\\
 &= H(S_U \; \mathlarger{\mid} \;  Z^n|_G,  Z^n|_{\tilde{G}_S}) + H(S_{\overline{U}} \; \mathlarger{\mid}  \; S_U, Z^n|_G,  Z^n|_{\tilde{G}_S})\\
 &\geq H(S_U \; \mathlarger{\mid}  \; Z^n|_G,  Z^n|_{\tilde{G}_S}) + H(S_{\overline{U}} \; \mathlarger{\mid}  \; X^n|_G,  X^n|_{\tilde{G}_S}).
\end{align*}
For the second term, we may directly invoke Lemma~\ref{lem:minor} using the fact that $|B|>|S_{\overline{U}}|=nr$ to conclude that with high probability the second term approaches $nr$. For the first term, let $\Upsilon$ be the typical event that the fraction of erasures in $(Z^n|_G,Z^n|_{\tilde{G}_S})$ is at least a $\epsilon_2(1-\frac{\delta}{2})$. 
\begin{align*}
 H(S_U \; \mathlarger{\mid}  \; Z^n|_G,  Z^n|_{\tilde{G}_S}) \geq  H(S_U \; \mathlarger{\mid}  \; Z^n|_G,  Z^n|_{\tilde{G}_S},\Upsilon)P(\Upsilon).
\end{align*}
Using the fact that $P(\Upsilon)\rightarrow 1$ and invoking Lemma~\ref{lem:minor} we may conclude that the first term also approaches $nr$ with high probability.

We may directly invoke Lemma~\ref{lem:minor} as we did above for the second term to conclude that $H(S_{\overline{U}} | X^n|_{\tilde{G}_S})\rightarrow nr$ with high probability. Hence, by union bound we can conclude that with high probability $H(S_U,S_{\overline{U}} \; \mathlarger{\mid}  \; Z^n|_G,  Z^n|_{\tilde{G}_S} )$ and $H(S_{\overline{U}} \; \mathlarger{\mid}  \; X^n|_{\tilde{G}_S})$ approach $2nr$ and $nr$, respectively.
\end{proof}

\begin{lemma} \label{lem:minor}
Let $\alpha, \beta, \delta > 0$ be such that 
$\alpha +\beta + \delta <1$. Let
$X$ be a vector chosen uniformly at random from $\{0,1\}^n$. Let
$f:\{0,1\}^n \rightarrow \{0,1\}^{n\alpha}$ be a map chosen uniformly at random
from all possible maps. Then with high probability over the choice of $f$,
the map satisfies the property that for every $I\subset \{1,2,\cdots,n\}$
with $|I| \leq n\beta$, and for every $y\in \{0,1\}^{|I|}$, 
$H(f(X)|X|_I=y) \geq n\alpha - 2^{-\delta n}$.
\end{lemma}
\begin{proof}
Let us first fix a particular subset $I$ and a realization
$X|_I=y$. Without loss of generality, we assume $|I|=n\beta$.
Let us denote $J:=I^C$ for simplicity. 
Now, $f:X|_{J} \mapsto \{0,1\}^{n\alpha}$ is only a function
of the components of $X$ in $J$. Let $Y_1,Y_2,\cdots, Y_N$ denote the
images $f(X|_{J},X|_I=y)$ of all $X|_{J}\in \{0,1\}^{|J|}$, where $N=2^{|J|}
=2^{n(1-\beta)}$.
Clearly these are independent and uniformly distributed over 
the $2^{n\alpha}$ binary strings in $\{0,1\}^{n\alpha}$.
The empirical distribution of $Y_i; \, i=1,2,\cdots, N$ is denoted as 
$\hat{p}_J$. Rest of the proof is exactly the same as that
of \cite[Lemma~10]{MishraDPDisit14}. We repeat it here for completeness. 

By Sanov's theorem,
\begin{align*}
Pr \left[H(\hat{p}_J) < n\alpha - 2^{-n\delta}\right] & \leq (N+1)^{2^{n\alpha}} 2^{-ND(p^*||u)}
\end{align*}
where $u$ dentoes the uniform distribution over $\{0,1\}^{n\alpha}$, and
\begin{align*}
p^* = \arg\min_{p:H(p) < n\alpha - 2^{-n\delta}} D(p||u). &
\end{align*}
Clearly,
\begin{align*}
D(p^*||u) = n\alpha - H(p^*) > 2^{-n\delta}. &
\end{align*}
So
\begin{align*}
Pr \left[H(\hat{p}_J)< n\alpha - 2^{-n\delta}\right] & < (2^{|J|} +1)^{2^{n\alpha}} 2^{-2^{|J|}\cdot 2^{-n\delta}} \\
& < 2^{n(1-\beta +1/n)\cdot 2^{n\alpha}} \cdot 2^{-2^{n(1-\beta)}\cdot 2^{-n\delta}} \\
& \leq 2^{-2^{n\alpha}(2^{n(1-\beta -\alpha -\delta)}-n(1-\beta +1/n))}.
\end{align*}
Since $\beta+\alpha+\delta <1$, by union bound, we have
\begin{align*}
& Pr \left[H(\hat{p}_J)< n\alpha - 2^{-n\delta} \text{ for some } I
\text{ and some } y\in \{0,1\}^{|I|}\right]   \\
& \hspace*{10mm} \leq 2^{-2^{n\alpha}}
\end{align*}
This goes to zero doubly exponentially fast as $n\rightarrow \infty$.
\end{proof}

\end{document}